\documentclass{article}
\usepackage[utf8]{inputenc}
\usepackage{amsmath}
\usepackage{amssymb}
\usepackage{graphicx}
\usepackage{amsthm}
\usepackage{latexsym}
\newtheorem{theorem}{Theorem}
\newtheorem{corollary}{Corollary}
\newtheorem{definition}{Definition}
\newtheorem{prop}{Proposition}
\newtheorem{lemma}{Lemma}

\newtheorem{remark}{Remark}
\usepackage{algpseudocode}
\usepackage{algorithmicx}
\usepackage{algorithm}

\newcommand{\beq}{\begin{equation}}
\newcommand{\eeq}{\end{equation}}
\newcommand{\barr}{\left[\begin{array}}
\newcommand{\earr}{\end{array}\right]}

\newcommand{\im}{\mbox{im\,}}
\newcommand{\rank}{\mbox{rank}\,}
\newcommand{\fp}{\mathbb{F}_{p}}
\newcommand{\fq}{\mathbb{F}_{q}}

\newcommand{\fqm}{\mathbb{F}_{q^{m}}}

\newcommand{\bpf}{\begin{proof}}
\newcommand{\epf}{\end{proof}}

\newcommand{\zz}{\ensuremath{\mathbb{Z}}}

\newcommand{\ftwo}{\ensuremath{\mathbb{F}_{2}}}

\newcommand{\ff}{\ensuremath{\mathbb{F}}}

\newcommand{\bi}{\begin{itemize}}
\newcommand{\ei}{\end{itemize}}
\newcommand{\bnum}{\begin{enumerate}}
\newcommand{\enum}{\end{enumerate}}
\newcommand{\bc}{\begin{center}}

\newcommand{\al}{\alpha}

\title{Local inversion of maps: A new attack on Symmetric encryption, RSA and ECDLP}
\author{Virendra Sule\\Dept.\ of Electrical Engineering\\Indian Institute of Technology Bombay, India\\(vrs@ee.iitb.ac.in)}
\date{January 14, 2022}
\begin{document}
\maketitle
\emph{Subject Classification}: cs.CC, cs.CR, math.NT
\begin{abstract}
    This paper presents algorithms for local inversion of maps and shows how several important computational problems such as cryptanalysis of symmetric encryption algorithms, RSA algorithm and solving the elliptic curve discrete log problem (ECDLP) can be addressed as local inversion problems. The methodology is termed as the \emph{Local Inversion Attack}. It utilizes the concept of \emph{Linear Complexity} (LC) of a recurrence sequence generated by the map defined by the cryptanalysis problem and the given data. It is shown that when the LC of the recurrence is bounded by a bound of polynomial order in the bit length of the input to the map, the local inversion can be accomplished in polynomial time. Hence an incomplete local inversion algorithm which searches a solution within a specified bound on computation can estimate the density of weak cases of cryptanalysis defined by such data causing low LC. Such cases can happen accidentally but cannot be avoided in practice and are fatal insecurity flaws of cryptographic primitives which are wrongly assumed to be secure on the basis of exponential average case complexity. An incomplete algorithm is proposed for solving problems such as key recovery of symmetric encryption algorithms, decryption of RSA ciphertext without factoring the modulus, decrypting any ciphertext of RSA given one plaintext ciphertext pair created with same public key in chosen ciphertext attack and solving the discrete logarithm on elliptic curves over finite fields (ECDLP) as local inversion problems. It is shown that when the LCs of the respective recurrences for given data are small, solutions of these problems are possible in practically feasible time and memory resources.
\end{abstract}

\section{Introduction}
This paper proposes a new attack called \emph{Local Inversion Attack} which is applicable for cryptanalysis of symmetric encryption algorithms (block and stream ciphers), RSA encryption and solution of the Elliptic Curve Discrete Log Problem (ECDLP). The attack is based on computation of a unique solution to inverting a function when the given value of the function defines a periodic recurring sequence of low \emph{linear complexity} (LC). The inverses of respective functions solve the symmetric keys in case of symmetric encryption given a known plaintext, plaintext input of the RSA ciphertext without resorting to factoring the modulus, decrypt every encryption of RSA by same public key when one pair $(c,m)$ is known and the discrete log $m$ in the ECDLP given the points $P$ and $[m]P$ on the elliptic curve. The sequences produced by the functions in respective cases have not been studied previously for their LCs. Hence an accidental low complexity values of the sequence in a specific session can be fatal for security of these primitives. The local inversion problem was recently posed and solved by the author in \cite{suleLI1,suleLI2} for maps in binary field case $\ftwo$ by proposing a complete algorithm for the solution which computes all solutions to the inversion. It was shown that the complete algorithm for local inversion involved solving NP-hard problems in general. An incomplete algorithm is also proposed in these references which is promising for solving the local inversion problem in the case when the LC is low. In this paper we extend this incomplete algorithm to general finite fields and explain the above important cryptanalysis applications.  

If $\ff$ is a finite field and $F:\ff^n\rightarrow\ff^n$ is a map. Then given a $y$ in $\ff^n$ in the image of $F$ local inversion of $F$ at $y$ is concerned with finding all $x$ in $\ff^n$ such that $y=F(x)$. An approach to solving this problem was proposed by the author in \cite{suleLI1} for binary field $\ftwo$. Some consequences and an improved incomplete algorithm were announced in \cite{suleLI2}. Previously well known approach for local inversion of maps has been known as TMTO attack which has been well known in cryptanalysis since long \cite{Hell}. However TMTO relies on the collision probability and does not use any properties of finite fields. In this paper we explore wider consequences of this methodology of incomplete algorithm for local inversion and show that formulation of cryptanalysis problems as local inversion gives a remarkable advantage for cryptanalysis in all the above problems of cryptanalysis. The methodology of local inversion by the incomplete algorithm utilizes the well known notion of LC of a sequence generated by the map $F$ recursively for the given value $y$ to find one solution in polynomial time if the LC happens to be of polynomial order $O(n^k)$. In most cases usually it is known that there is a unique solution to the problem hence the inverse obtained is the only solution. Hence this methodology offers a solution in practically feasible time and memory for sufficiently small $k$. However such solutions are only feasible for special cases of data $y$ as an output of the map $F$ which corresponds to a specific session of the primitive. Such low complexity cases of local inversion given $y$ are accidental and possibly of very low density among all possible operational values of the map $F$ or what can be called as average cases of the primitive. Unfortunately as encryption may be carried out using the same private or symmetric key on multiple inputs any one of input producing an output $y$ of low LC shall lead to a fatal breaking of the encryption for all inputs. Hence even if such low complexity cases are rare, once a case is found the private key as well as all encryptions by the key are compromised. Also the probability of accidental occurrence of low complexity cases increases with practical feasibility bounds (on time and memory) allotted for local inversion algorithm. These special cases spell doom on the whole encryption if their density is not exponentially small given these practical bounds. Hence security of encryption based on exponential average case complexity is no more a reliable assurance of security.

\subsection{Relationship with previous work}
The TMTO attack algorithm of \cite{Hell} addressed the problem of local inversion of a map $F:\{0,1\}^n\rightarrow\{0,1\}^m$ for $n\leq m$. The attack still attracts newer research as in \cite{hongsarkar, hays, gango} even after forty years. TMTO however does not utilize the structure of finite fields defining the domain and range of the map. TMTO is essentially based on collision probability and hence the complexity of TMTO is of the order of $O(\sqrt{N})$ where $N=2^n$. Local inversion algorithms of \cite{suleLI1,suleLI2} and those to be presented in this paper on the other hand depend on the structure of the finite field $\ff$ inherent in definition of the map $F:\ff^n\rightarrow\ff^n$. The complexity of local inversion is in terms of $\log N$ where polynomial time is of the order $O((\log N)^k)$. Hence the TMTO is an exponential algorithm and is of a completely different nature of search. The incomplete algorithm for local inversion presented in this paper sets a bound on computation in $O((\log N)^k)$ and determines if there is a solution to the inversion within the bound. Hence for sufficiently small bound such as $k=3$ the computation is equivalent to solving a linear system of size $O(n^k)$ for a reasonably small $k$. This methodology is useful to make estimates of the density of small complexity encryptions (which are weak) and which can be broken in time and memory estimates proportional to solving a linear system of size $O(n^k)$. 

\subsubsection{Local Inversion attacks on RSA}
Cryptanalysis of RSA has always been studied from the viewpoint of factoring the modulus. While major progress has taken place in factorization algorithms the best known of these have sub-exponential complexity in the bit length $l=\log n+1$ of the modulus $n=pq$ where $p,q$ are distinct odd primes \cite{Lenstra,Lenstraetal,Buhleretal}. Local inversion approach to cryptanalysis of RSA on the other hand is a completely different problem of inverting the maps $y=F(x)$. The maps $F$ can be defined in two different situations.
\begin{enumerate}
    \item Inversion of ciphertext $c$ to plaintext $m$, $c=m^e\mod n$. Here the map to be inverted is $c=F(m)$, hence maps $\ftwo^l\rightarrow\ftwo^l$.
    \item Inversion in Chosen Ciphertext Attack (CCA) to find the private key $d$ since $m=c^d\mod n$ assuming both $m,c$ known. Here $F:\ftwo^{\log\phi(n)}\rightarrow\ftwo^l$ since the unknown $d$ belongs to $\zz_{\phi(n)}$. Hence this is the case of local map inversion $F:\ftwo^{(l-1)}\rightarrow\ftwo^l$ which is an embedding.
\end{enumerate}
The incomplete algorithm of local inversion of $c$ to $m$ then determines whether there exist low complexity encryptions $y$ of polynomial order $O(l^k)$ and particularly of practically feasible order $O(l^3)$ and how much size of linear system would be required to break such weak encryptions. Clearly this method of inverting RSA has no theoretical overlap with factorization. However the complexity $O(l^k)$ shall be a concrete comparison with order of complexity of factoring $l$ bit length moduli $n=pq$. The embedding map inversion case also arises in the discrete log computation as discussed next.

\subsubsection{Solution of the ECDLP}
Among the most important problems of cryptanalysis which can be studied using the local inversion methodology is the ECDLP. Here the map $F$ is defined by the relation $Q=[m]P$ where $P,Q$ are given points on the elliptic curve $E(\ff_q)$ and $m$ the unknown  multiplier number. $F$ maps $m$ to $E$. Hence this is not the standard problem of local inversion since bit length of the exponent $m$ the unknown to be solved is less than the bits required for representing a point on $E$. It is assumed that the point $Q$ is in the cyclic subgroup $<P>$. The order of the subgroup is bound by the order of $E$ which is further bound by $q+1+2\sqrt{q}$. Since as a set of points $E$ is embedded in $\ff_q^2$ and $r=q+1+2\sqrt{q}<2q$ the map $F$ can be formulated in bits as a map  $F:\ff^{\log r+1}\rightarrow\ff^{\log(q)+2}$ which is an embedding $F:\ff^n\rightarrow\ff^m$, $n<m$. The complexity of local inversion of such a map as we shall show in a later section is that of simultaneous local inversion of $m-n$ maps $F_i:\ff^n\rightarrow\ff^n,i=1,\ldots(m-n)$ defined by $F$ with a common inverse. Hence the incomplete algorithm allows us to determine the number of weak cases of ECDLP of polynomial size complexity $O(n^k)$. Such a methodology does not seem to have appeared in any of the previous methods of solving the ECDLP. Even the most recent analogous index calculus method of solving the ECDLP has exponential complexity \cite{SATECDLP} despite giving much improvement in computation. Due to wide application of Elliptic Curve Cryptography (ECC) on internet, private networks and the digital currency the solution to ECDLP using the local inversion methodology shall be an important new approach in estimation of security and identification of weak sessions of key exchange in ECC. 

\subsubsection{Differences with LC of stream generators and the Berlekamp Massey attack}
Stream ciphers and block ciphers in counter mode produce periodic pseudorandom (PR) streams with long periods of exponential order in the number $n$ of bits of the input secret key or the seed. Berlekamp Massey (BM) attack \cite{StampLow} on such streams is to model the stream as an output of a linear feedback shift register (LFSR) with a known initial loading thereby allowing recovery of the entire stream used for encryption. BM attack succeeds if the LC (the degree of the minimal polynomial) of the stream is of low order (polynomial order in $n$). A large class of sequences have been discovered in the literature \cite{Nied1, Fang} which have large linear complexities (close to half the period which is of exponential order). However such studies are useful to show if a particular sequence has large LC. In practice when such streams are generated by stream generators each stream is uniquely identified with a secret key $K$ and an initialising seed called $IV$ (which is publicly known). Hence even if BM attack does succeed for one stream (for a known IV) it does not reveal the secret key. Usually the stream generators are designed and $IV$s chosen such that the LC is not small and the attack is defeated. The $k$-linear complexity \cite{Nied2} is another useful concept which shows variation of LC of a sequence relative to close sequences differing in few bits

However, it is important to note that the local inversion attack we propose in this paper is not concerned with the LC of the output stream of a stream generator at all. The attack considers a totally different sequence recursively generated by the map from secret key to the output stream. Hence none of the previous studies of LC of output stream of generators (for specific $IV$) are relevant to determine the LC of the recurrence sequence utilized for local inversion attack. Moreover once the LC of the recurrence is of polynomial order in the key length, the local inversion attack solves for the secret key in polynomial time. Hence the cipher algorithm is broken for all $IV$s which may be used as seeds for the same key. Hence this attack is of very different kind than the BM attack of modeling the output stream by an LFSR. This new attack must be considered as an additional measure of security of PR generators apart from the BM attack on the output streams. 

\subsubsection{Local inversion as a general problem over finite fields}
Inversion of maps is a problem of fundamental interest to Mathematics and Computation and has vast application in Physical, Life and Social Sciences. For instance local inversion can be applied to the problem of solving (rational solutions) of systems of polynomials over finite fields. Inversion of maps and functions has been studied in depth for over a century over continuous domains while over discrete domains such as finite fields, importance and applications of the inversion problem have been slow to come to light. Cryptanalysis happens to be one such vitally important application in which almost all problems are known to be computationally hard. Applications of maps over finite fields arise in several applications such as Biological and Chemical networks and symbolic dynamical systems. In number theory, the map such as the exponential function modulo a prime and in finite fields does not appear to have been studied from the point of view of LC of the recurrence. LC is a characteristic property of primes and generating polynomials of field extensions. Dynamical systems and maps in finite fields exhibit interesting linear representation and its application as shown in \cite{Ramsule1,Ramsule2}. These show that the theory of inversion of maps in finite and discrete domains has important real life applications which are yet to be fully explored. 

\subsection{Organization of the paper}
A brief background of the local inverse problem and its computational solution by an incomplete algorithm is first presented in Section 2. The main difference of material in this paper compared to the previous papers \cite{suleLI1,suleLI2} is the focus on incomplete algorithm and applicability of the theory to general finite fields. Focus of the previous paper \cite{suleLI1} was primarily on the complete algorithm to find all solutions. The complete algorithm is at present possible only for maps over the binary fields $\ftwo$ and involves solving NP hard problems to obtain all solutions. The incomplete algorithm only requires computing a single solution in a periodic sequence and is accomplished by solving a linear system over the field and hence works over any finite field.

The case of embedding maps is then analysed in Section 3 and the incomplete algorithm is extended to solve the problem of local inversion on an embedding.

In section 4, the cases of cryptanalysis of symmetric key algorithms of block and stream ciphers are discussed in terms of local inversion problem using the incomplete algorithm. 

In Section 5, an application of local inversion is presented for cryptanalysis of RSA which shows that RSA encryption can be decrypted by local inversion of an appropriate map without factoring the modulus. It is also shown how in one CCA attack, local inversion leads to breaking of all the encryptions by the same private key. Bounds on the sizes of linear systems to be solved for polynomial order $O(l^k)$ bounds on the complexity are shown where $l$ is the bit length of the modulus. 

Section 6 is devoted to the formulation of the ECDLP over finite fields as a local inversion problem of an appropriate embedding map $F:\ftwo^n\rightarrow\ftwo^m$ over the binary field. Bounds on sizes of linear systems required to be solved if the linear complexity of sequences defined by the map are bounded by polynomial order $O(n^k)$ are presented. Solving DLP on finite fields is then presented as a local inversion problem and appropriate maps to be inverted are derived.

This paper shall focus on the presentation of the theory of incomplete algorithm for local inversion of maps in three central problems of cryptanalysis as discussed above. Computational details and measurement of densities of low LC cases of cryptanalysis problems shall be presented in separate articles.

\section{Local Inversion by an incomplete algorithm}
In this section we gather the background required for presenting the theory. Although there is much overlap of this section with the material in \cite{suleLI1,suleLI2} the case here is over a general finite field while that in these previous papers was for the binary field $\ftwo$. Moreover a recall of the basic background is more convenient for readers instead of referencing back. Consider a map $F:\ff^n\rightarrow\ff^n$ where $\ff$ is a finite field. Given $y$ in $\ff^n$ local inversion problem is to determine all solutions $x$ in $\ff^n$ of the equation $y=F(x)$. A complete algorithm for solving this problem over the binary field $\ftwo$ is presented in \cite{suleLI1,suleLI2} which returns the set of all solutions of the problem (or returns an empty set if there is no solution). An incomplete algorithm takes as an input a bound $M$ on the computation (which bounds number of steps of time and size of memory) and returns a unique solution if any solution is found within this bound on computation. Algorithms for solution to the inversion problem are based on the structure of solutions of the equation 
\beq\label{eqn:fundamental}
y=F(x): x,y\in\ff
\eeq
over a finite field $\ff$. The structure of solutions arises from the dynamical system over $\ff^n$ defined by the map $F$ as its transition function. 

\subsection{Dynamical system defined by $F$ and structure of solutions}
The map $F$ in the equation (\ref{eqn:fundamental}) defines the transition function of the dynamical system
\beq\label{eqn:dynsys}
x(k+1)=F(x(k)),k=0,1,2,\ldots
\eeq
with state $x(k)$ in $\ff^n$. An initial state $x(0)$ defines a unique trajectory through $x(0)$ of the system (\ref{eqn:dynsys}) in $\ff^n$ and also generates the sequence of recurrences (or iterates of $F$)
\[
S(F,x(0))=\{x(0),F(x(0)),F^{(2)}(x(0)),\ldots\}
\]
Hence the system (\ref{eqn:dynsys}) and the sequence $S(F,x(0))$ are frequently referred interchangeably once $x(0)$ is specified. Due to finiteness of the domain $\ff^n$ of the transition map, the trajectories of the dynamical system are limited to only three types, fixed points, periodic points and chains.
\begin{enumerate}
    \item \emph{Fixed points}: Points $x$ in $\ff^n$ which satisfy, $F(x)=x$.
    \item \emph{Periodic points of period $N$}: Points $x$ in $\ff^n$ which satisfy $F^{(N)}(x)=x$. (Hence fixed points are periodic points of period $1$). The trajectory $x(k)$ with initial state $x$ is called a closed orbit of length $N$ through $x$. The trajectory is aslo termed a \emph{periodic orbit} of period $N$.
    \item \emph{Chains of length $l$}: Trajectories $x(k)$ with initial state $z$ in the set called Garden of Eden (GOE) of $F$,
    \[
    \mbox{GOE}=\{z\in\ff^n|z\notin\im F\}
    \]
    and the final state $x(l)$ in a closed orbit or a fixed point.
\end{enumerate}
Basic observations about trajectories of dynamical system (\ref{eqn:dynsys}) are
\begin{prop} Following statements hold
\begin{enumerate}
    \item Every trajectory $x(k)$ of a dynamical system (\ref{eqn:dynsys}) is one of the above three type. 
    \item If $x$, $y$ are distinct periodic points then either they are in the same periodic orbit or have disjoint (non-intersecting) periodic orbits.
    \item A map $F$ is a permutation iff all trajectories are closed orbits and its GOE is an empty set. Hence the closed orbits of a permutation partition $\ff^n$.
\end{enumerate}
\end{prop}
\begin{proof}
The proof of these statements follow just from the finiteness of $\ff^n$, because of which every trajectory $x(k)$ is \emph{ultimately periodic}, there exist $p$, $N$ for every $x$ such that the trajectory $x(k)$ with initial state $x$ satifies, $x(k+N)=x(k)$ for $k\geq p$. ($p$ is called the \emph{pre-period} of the trajectory and the smallest $N$ is called the \emph{period}). Hence a trajectory $x(k)$ if not already periodic always has a part in a chain before it reaches a periodic orbit.

Next let $x$, $y$ be distinct periodic points and not in eachother's closed orbit. If $w$ is a common point in their closed orbits, then the closed orbit of $w$ is in closed orbits of both $x$ and $y$ hence their orbits are the same. Hence it follows that closed orbits containing $x$ and $y$ do not intersect.

Finally, a map $F$ is a permutation then it is one to one in $\ff^n$ hence every $y$ has an inverse image $x$ such that $y=F(x)$ hence the GOE is empty. Since all closed orbits are distinct because of the previous statement and isolated points are fixed points these are all the points in $\ff^n$. Hence the trajectories of (\ref{eqn:dynsys}) partition $\ff^n$.
\end{proof}

Above Proposition describes the structure of solutions of the equation (\ref{eqn:fundamental}). First make easy observation
\begin{prop}
Equation $y=F(x)$ has no solution iff $y$ belongs to $GOE$ of $F$.
\end{prop}
Next, a theorem on the structure of all solutions of (\ref{eqn:fundamental}).

\begin{theorem}\label{strofsolns} Following statements hold
\begin{enumerate}
    \item $F(x)=y$ has a solution $x$ in a periodic orbit $P$ iff $y$ belongs to $P$. Such a periodic orbit and hence also the solution $x$ in $P$ is unique. 
    \item All other possible solutions belong to the chains $F^{k}(z),k\geq 1$ for $z$ in the GOE of $F$.
    \item If $y$ is neither in a periodic orbit nor in the GOE, then solutions arise in some of the segments $F^{(k)}(z),1<k\leq (l-1)$ for some $l$ for some of the $z$ in the GOE.
    \item $y=F(x)$ has a unique solution iff $y$ belongs exclusively to a chain segment of a unique point in GOE or to a unique periodic orbit 
\end{enumerate}
\end{theorem}

Proof of the theorem is presented in the appendix as it is almost same as that presented as proof of Lemma 1 in \cite{suleLI1} except for the fact that this proof is valid for any finite field.

\subsection{Computation of solutions}
From the theorem on structure of solutions of the equation (\ref{eqn:fundamental}) it follows that any solution of (\ref{eqn:fundamental}) over a finite field $\ff$ is either in one unique periodic orbit of trajectories of the dynamical system (\ref{eqn:dynsys}) or in a chain (a trajectory starting from some initial point in the GOE). It is shown in \cite{suleLI2} for the map $F$ over $\ftwo$ that computation of GOE is an NP-hard problem. Hence computation of solutions in chains is beyond the goal of an incomplete algorithm with polynomial size bounds on computation. Although this complexity is established only over $\ftwo$, the situation can be predicted to be even harder for general finite fields. In fact an algorithm for solving this problem of computation of GOE of a map $F$ on finite fields other than $\ftwo$ does not seem to have been discussed in previous literature.

\subsubsection{Computation of solution in a periodic orbit}
We now assume that $y$ in (\ref{eqn:fundamental}) belongs to a periodic orbit of the iterates of (\ref{eqn:dynsys}) or that the sequence $S(F,y)$ is periodic
\beq\label{Sequence}
S(F,y)=\{y,F(y),F^{(2)}(y),\ldots\}
\eeq
Let $N$ be the period of this sequence or the period of the trajectory with initial condition $y$. Hence $N$ is the smallest index such that
\[
F^{(N)}(y)=y
\]
Hence $x=F^{(N-1)}(y)$ is the solution of the equation. Thus computing the unique solution in the periodic orbit $S(F,y)$ it is sufficient to
\begin{enumerate}
\item Determine whether the sequence $S(F,y)$ (trajectory of $y$) is periodic.
    \item Compute the period $N$ of the tratejctory $S(F,y)$ of $y$.
    \item Evaluate the $(N-1)$-th compositional power of $F$ at $y$. 
\end{enumerate}
Even when $N$ is of exponential order in $n$, the compositional power $F^{(N-1)}(y)$ can be computed in polynomial time $O(n^k)$. This is possible by the repeated compositional squaring shown for maps over $\ftwo$ in \cite{suleLI1} (appendix) analogous to repeated squaring in groups to compute powers of exponential order. Same proof can be extended for any finite field. Hence to find the solution $x$ it is required to estimate complexities of determining whether the sequence $S(F,y)$ is periodic and computing the period $N$. 

\subsubsection{The practical hurdle in computation}
We now highlight the computational hurdle in computing the period. In practice the numbers of terms of the sequence $S(F,y)$ available for computation are limited (polynomial size $O(n^k)$ for sufficiently small size $k$ as is determined by practical feasibility). Hence the above problems of determining periodicity of $S(F,y)$ and computing the period $N$ are required to be solved for a subsequence of $S(F,y)$ with a limited number of terms. Both of these problems are linked with the concept of linear complexity of $S(F,y)$. This is discussed next.

\subsection{Linear complexity and the solution of $y=F(x)$}
It turns out that for computing the solution $x$ of (\ref{eqn:fundamental}) it is not necessary to compute the sequence terms upto $(N-1)$ for a period $N$ and the compositional power $F^{(N-1)}(y)$. Instead, it is sufficient to discover the \emph{Linear Complexity} (LC) and the \emph{minimal polynomial} of the sequence $S(F,y)$ from the limited terms of the sequence. A greatest advantage in such an alternative as we show below is that if the LC of $S(F,y)$ is of polynomial order $O(n^k)$ then it can be computed along with the minimal polynomial in polynomial time and this allows computation of the solution $x$ in polynomial time directly from the minimal polynomial. Alternatively it is shown that when LC is of polynomial order, the period $N$ can be computed in polynomial time as the order of the minimal polynomial. Moreover the compositional power $F^{(N-1)}(y)$ can also be found in polynomial time which is the solution of the inverse. Hence then these two ways of computing the solution are computationally polynomial time equivalent. However the direct computation of $x$ from the minimal polynomial shown and used in the algorithm is practically much advantageous. This is the big picture of the idea of this paper.

\subsubsection{Linear recurrence relation satisfied by $S(F,y)$}
To describe the above observation it is necessary to recall the concept of linear recurrence relation satisfied by a periodic sequence over a finite field $\fq$ and a polynomial associated with such a relation. Linear recurrence relations in sequences have been studied since long. The concept of \emph{linear complexity} (LC) and the modeling of a sequence by an LFSR of length equal to LC using the Berlekamp-Massey algorithm have been known in the theory of stream ciphers \cite{Rueppel, StampLow}. We refer the reader to \cite{Gologong} for an immediate relevance to periodic sequences while an exposition of a general theory is available in \cite{LidlNied}. Although linear recurrence and minimal polynomials of periodic sequences have been studied since long their application for local inversion problem does not seem to have caught attention of the researchers in the past as evident from a near total absence of reference to this problem. While inversion of a map is an objective of TMTO attack, the TMTO algorithm does not utilize the structure of finite field of the domain hence none of the mathematical constructs such as minimal polynomial are applicable for TMTO attack.

\subsection{Minimal polynomial, order and period of $S(F,y)$}
The sequence $S(F,y)$ is said to satisfy a recurrence relation if there exist $j_0\geq 0$, $m\geq 1$ and constants $\al_0,\al_1,\ldots,\al_{(m-1)},\al_m\neq 0$ such that the following relations hold
\beq\label{LRR}
\al_mF^{(k+m)}(y)=\sum_{i=0}^{(m-1)}\al_iF^{(k+i)}(y)
\eeq
for $k=j_0,j_0+1,j_0+2,\ldots$. If such a recurrence relation is found, a polynomial $p(X)$ can be associated with the relation as follows
\[
p(X)=\al_mX^m-\sum_{i=0}^{(m-1)}\al_iX^{i}
\]
When $\al_m\neq 0$, the degree $m$ of the polynomial is also called the \emph{degree} of the recurrence relation (\ref{LRR}). The polynomial 
\[
\phi(X)=(1/\al_m)p(X)
\]
is called a \emph{characteristic polynomial} of $S(F,y)$. We can also say that the recurrence relation (\ref{LRR}) with $\al_m=1$ is defined by the polynomial $\phi(X)$. (Hereafter we shall drop the term linear and simply call (\ref{LRR}) as a recurrence relation). A characteristic polynomial of smallest degree $m$ denoted $m(X)$ is unique, for if not, (\ref{LRR}) shows that there is a recurrence relation of degree smaller than degree $m$. Hence the polynomial is called the \emph{minimal polynomial} of $S(F,y)$, which is monic and defines the least degree recurrence relation. We shall denote the minimal polynomial $m(X)$ as follows, retaining the notation for the co-efficients $\al_i$ for $i\geq 0$ as
\beq\label{minpoly}
m(X)=X^m-\sum_{i=0}^{m-1}\al_0X^i
\eeq
Formally we can denote the indeterminate $X$ and polynomials in $\ff[X]$ as linear operations on the sequence $S(F,y)$ induced by the rules,
\[
X^k(y)=F^{(k)}(y),\mbox{ }(aX^k+bX^l)(y)=aF^{(k)}(y)+bF^{(l)}(y),a,b\in\ff
\]
Then it follows that any characteristic polynomial $\phi(X)$ satisfies
\[
\phi(X)(F^{(k)}(y))=0\forall k\geq 0
\]
which is denoted in short as $\phi(S(F,y))=0$. An important characterization of the minimal polynomial is given by the following well known Proposition \cite{Gologong, suleLI1}.

\subsubsection{Relations between order and period}
For a polynomial $f(X)$ in $\ff[X]$ which satisfies $f(0)\neq 0$ there exists a number $N$ which is the smallest number such that $f(X)|(X^N-1)$ ($f(X)$ divides $(X^N-1)$). This number is called the \emph{order} of $f(X)$ over $\ff$, denoted
    \[
    N=\mbox{order}\;f(X)
    \]
Order of a polynomial has well known relation with the orders of roots of its irreducible factors over its splitting field \cite{LidlNied}.

\begin{prop}\label{Existenceofmp}
The sequence $S(F,y)$ is periodic iff it has a minimal polynomial which divides any of its characteristic polynomials. The minimal polynomial satisfies $\al_0\neq 0$ (equivalently $m(0)\neq 0$) and the period $N$ of $S(F,y)$ is the order of the minimal polynomial.
\end{prop}

The proof is given in \cite{suleLI1} as proof of Proposition 1 hence is omitted. It is stated there only over the field $\ftwo$ however it works for any finite field with very minor modifications.

\subsection{Solution of the equation $y=F(x)$}
An important observation relating the minimal polynomial of a periodic $S(F,y)$ to the unique solution $x$ of $\ref{eqn:fundamental}$ is given by

\begin{theorem}\label{Solutioninperiodicorbit}
Let $S(F,y)$ be a periodic sequence and $m(X)$ as described in (\ref{minpoly}) be its minimal polynomial. Then there is a unique solution to $F(x)=y$ in $S(F,y)$ given by
\beq\label{solution}
x=(1/\al_0)[F^{(m-1)}(y)-\sum_{i=1}^{(m-1)}\al_iF^{(i-1)}(y)]
\eeq
\end{theorem}
Proof is presented in the appendix. It is essentially a reproduction of the proof of Theorem 1 stated in \cite{suleLI1}. While it was stated earlier only for the field $\ftwo$, present proof works for any $\ff$.

\subsubsection{Linear complexity and computation of the minimal polynomial}
From Theorem (\ref{Solutioninperiodicorbit}) it follows that the unique solution $x$ of (\ref{eqn:fundamental}) in the periodic orbit containing $y$ is obtained once the unique minimal polynomial $m(X)$ is computed from the sequence $S(F,y)$. We first state the theoretical result in this connection as stated in \cite{suleLI1} for general fields and then take the problem of constructing an incomplete algorithm to set up the computation to solve for $m(X)$ in the practical situation when only a partial sequence $S(F,y)$ is specified.

A well known algorithm for computation of the minimal polynomial of a sequence is the Berlekamp-Massey algorithm. Its role in the cryptanalysis of stream ciphers has been well known \cite{Rueppel}. However we shall present the computation in terms of the linear system involving the Hankel matrix defined by the sequence $S(F,y)$. If an arbitrary number of terms in the sequence are available the Hankel matrix of size $m$ starting from the term $y$ is given by
\beq\label{hankelmatrix}
H_{m+j}=
\barr{llll}
y & F(y) & \ldots & F^{(j+(m-1))}(y)\\
F(y) & F^{(2)}(y) & \ldots & F^{(j+m)}(y)\\
\dots & \vdots &  & \vdots\\
F^{(j+(m-1))}(y) & F^{(j+m)}(y) &\ldots & F^{(j+(2m-2))}(y)
\earr
\eeq
Following proposition then gives a criterion and method of computation of the minimal polynomial of $S(F,y)$. This is reproduced from \cite{suleLI1}. The proof is omitted but the same proof is valid for any field $\ff$ (while it is considered only for $\ftwo$ in the previous paper).

\begin{prop}\label{Prop:Hankel}
Let $S(F,y)$ be periodic then it has a minimal polynomial of degree $m$ iff
\beq\label{rankcond}
\rank H(m+j)=\rank H(m)=m
\eeq
for all $j=1,2,\dots$. The co-efficient vector 
\[
\hat{\alpha}=(\al_0,\ldots,\al_{(m-1)})^T
\]
of the minimal polynomial is the unique solution of
\beq\label{Hankeleqn}
H(m)\hat{\alpha}=h(m+1)
\eeq
where
\[
h(m+1)=[F^m(y),F^{(m+1)}(y),\ldots,F^{(2m-1)}(y)]^T
\]
\end{prop}

\subsection{An incomplete algorithm for computing a solution}
Now we come to the practical issue of computing the minimal polynomial. In practice the sequence $S(F,y)$ can never be specified completely over one period, because its period $N$ is of exponential order in $n$. Hence we are really not sure whether the minimal polynomial computed from equation (\ref{Hankeleqn}) by checking the condition (\ref{rankcond}) at $m$ correctly represents the complete sequence $S(F,y)$. However when the minimal polynomial is correct to represent this sequence then the inverse $x$ computed using the formula (\ref{solution}) is the correct local inverse of $y$ which is same as saying that $x$ verifies $y=F(x)$. This is the basis of the incomplete algorithm to find the local inverse in the periodic orbit of $y$ described below. 

\begin{algorithm}\label{IncompleteAlgo}
\caption{Incomplete algorithm to find the unique solution in periodic orbit}
\begin{algorithmic}[1]
\State\textbf{Input}: $y$ in $\ff^n$, $M$ an upper bound of polynomial order $O(n^r)$ (for length of the sequence $S(F,y)$) for $r$ decided based on practical feasibility.
\State\textbf{Output}: One solution of $F(x)=y$ in the periodic orbit $S(F,y)$ if one exists with minimal polynomial of degree $m\leq\lfloor M/2\rfloor$ otherwise returns that there is no conclusion.
\Procedure{Solution\_in}{$S(F,y)$}
\State Compute the sequence $S(F,y)$, $\{F^{(k)}(y)\}$ for $k=0,1,\ldots,M$.
\If {at any $k<M$, $F^{(k)}(y)=y$} \State Return $x=F^{(k-1)}(y)$ is the solution. 
\Else
\State Set $m=\lfloor M/2\rfloor$.
\Repeat
\State Compute matrices $H(m)$, $H(m+1)$ (as shown in (\ref{hankelmatrix})).
\State Compute $\rank H(m)$, $\rank H(m+1)$.
\If{
\[
m=\rank H(m)=\rank H(m+1)=m
\]
}
\State Compute the minimal polynomial co-efficients $\hat{\alpha}$ in (\ref{Hankeleqn}).
\State Compute solution $x$ as in (\ref{solution}).
\If {$F(x)=y$}
\State \% Verification to avoid false positive solutions.
\State Return: Solution $x$
\EndIf
\EndIf
\If{$\rank H(m)<\rank H(m+1)$}
\State Return: No conclusion within the bound $M$, ($\mbox{LC}>\lfloor M/2\rfloor$).
\EndIf
\If{$\rank H(m)=\rank H(m+1)<m$}
\State $m\leftarrow m-1$.
\EndIf
\Until{$m=1$}
\State Return: No conclusion within the bound $M$, ($\mbox{LC}>\lfloor M/2\rfloor$).
\EndIf
\EndProcedure
\end{algorithmic}
\end{algorithm}

We shall formally state the definition of minimal polynomial of $S(F,y)$ computed from a subsequence of $M$ terms.

\begin{definition}
Consider the subsequence of $S(F,y)$ given upto $M$ terms.
\[
\{y,F(y),F^{(2)}(y),\ldots,F^{(M-1)}(y)\}
\]
If the rank condition (\ref{rankcond}) holds at $m\leq\lfloor M/2\rfloor$ and the solution $x$ obtained in (\ref{solution}) satisfies $y=F(x)$ then the subsequence is said to have the minimal polynomial of degree $m$.
\end{definition}
 
\begin{theorem}
    If the sub-sequence $S(F,y)$ given up to $M$ terms where $M$ has polynomial order $O(n^k)$ has a minimal polynomial of degree $m\leq\lfloor M/2\rfloor$ then one solution $x$ of (\ref{eqn:fundamental}) can be computed in polynomial time
    \end{theorem}
Proof of the theorem follows immediately from the way Algorithm (\ref{IncompleteAlgo}) is constructed using the Proposition (\ref{Prop:Hankel}) and the formula of solution of the inverse $x$ given in (\ref{solution}).  The polynomial time assertion follows because the bound $M$ is of polynomial size in $n$ which makes the degree $m$ as well as number of terms $F^{(k)}(y)$ needed to be computed to find $x$ bounded by a polynomial order. 

\begin{remark}
The incomplete algorithm can be used in estimating the density of values $y$ in $\ff^n$ which result in a small (polynomial order) LC of the sequence $S(F,y)$. This density is the probabilistic estimate of the number of instances of $y$ for which computation of local inverse is feasible. Another variation of the incomplete algorithm is to progressively increase the degree $m$ of the minimal polynomial starting from a small degree $m_0$ at which the rank condition $\rank H(m_0)=\rank H(m_0+1)$ holds. Then verify whether solution $x$ obtained is correct. Increase $m$ to find the minimal polynomial and verify the solution $x$ until $m>\lfloor M/2\rfloor$. Many other variations of the algorithm are possible to exploit parallel computation of several possible minimal polynomials and possible solutions $x$ to verify. These discussions shall be a subject of a separate article on implementation of the incomplete algorithm. 
\end{remark}

This section completes the presentation of the ideas behind computing local inverse in practically feasible resources in time and memory. Applications of this approach using the incomplete algorithm for cryptanalysis are described in the following sections.

\section{Local Inversion of Embedding}
In the previous section the incomplete algorithm was proposed to solve the local inversion of maps $F:\ff^n\rightarrow\ff^n$. In many situations of cryptanalysis however, the map available is an embedding $F:\ff^n\rightarrow\ff^m$ for $n<m$. In this section we address the problem of local inversion of such embeddings and extend the incomplete algorithm to find the local inverse $x$ in a given equation $y=F(x)$.

\subsection{Simultaneous maps associated with embedding}
Clearly the difficulty in applying the previous theory of inversion of maps $F:\ff^n\rightarrow\ff^m$ when $n<m$ is that the recurring sequence $S(F,y)$ (\ref{Sequence}) or the dynamical system (\ref{eqn:dynsys}) are not defined. Following observation is useful in defining simultaneous maps associated with an embedding and a solution $x$. Let $t=m-n$, define projection maps 
\[
\Pi_i:\ff^n\rightarrow\ff^m
\]
for $i=1,2,\ldots (t+1)$ by 
\[
\Pi_i(x_1,x_2,\ldots,x_m)=(x_{1+i-1},x_{2+i-1},\ldots,x_{(n+i-1)})
\]
Then $\Pi_(1)=(x_1,x_2,\ldots,x_n)$, $\Pi_2=(x_2,x_3,\ldots,x_{(n+1)})$ and so on till $\Pi_{(t+1)}=(x_{(t+1)},x_{(t+2)},\ldots,x_m)$
\begin{lemma}
If an embedding $F:\ff^n\rightarrow\ff^m$ is given with $t=m-n>0$ then equation $y=F(x)$ has a solution $x$, iff the following equations simultaneously have solution $x$
\beq\label{projectioneqns}
y(i)=F_i(x),\mbox{for}\;i=1,2,\ldots (t+1)
\eeq
where $y(i)=\Pi_i(y)$ and $F_i=P_i\circ F$ for $i=1,2,\ldots,(t+1)$ which are maps $F_i:\ff^n\rightarrow\ff^n$, 
\end{lemma}

\begin{proof}
If $y=F(x)$ is satisfied then $\Pi_i(y)=y(i)=\Pi_i\circ F(x)=F_i(x)$. Hence the necessity is obvious.

Conversely if the system of equations (\ref{projectioneqns}) (called as projection equations) are satisfied by $x$ simultaneously, then for each co-ordinate of $y$ the equation
\[
y_j=f_j(x)
\]
holds for $j=1,\ldots,n$ from the first projection equation (\ref{projectioneqns}) for $i=1$. Then from projection equation for $i=2$, $y_{(n+1)}=f_{(n+1)}(x)$, from projection equation for $i=3$, $y_{(n+2)}=f_{(n+2)}(x)$ and subsequently from further projection equations, $y_{(n+j)}=f_{(n+j)}(x)$ are satisfied until $j=t$. Hence for all co-ordinates of $y$ the equations satisfied are
\[
y_{i}=f_{i}(x)\;\mbox{for}\;i=1,\ldots,m
\]
Hence $x$ satisfies the equation $y=F(x)$.
\end{proof}

The individual projection equations (\ref{projectioneqns}) have maps $F_{i}:\ff^n\rightarrow \ff^n$ hence the incomplete algorithm of the previous section can be employed to find the solution if and when one of them $S(F_i,y(i))$ is a periodic sequences. Hence we can now write the following

\begin{theorem}
If the any one of the sequences $S(F_i,y(i))$ for $i=1,2,\ldots, (t+1)$ is periodic and has LC of polynomial order $O(n^k)$, then a unique solution $x$ of the equation $y=F(x)$ for the given embedding exists iff $x$ satisfies each of the equations (\ref{projectioneqns}). The solution can be computed in polynomial time.
\end{theorem}

Despite the theorem asserting a solution $x$ iff all the systems (\ref{projectioneqns}) have this solution the theorem has the same limitation as before that the minimal polynomial of the sequence $S(F_i,y(i))$ can never be known in practice since only a limited polynomial size terms of the sequence can be computed. Hence only an incomplete algorithm can detect possible solutions. We extend the previous incomplete algorithm to the present case of solving the equation $y=F(x)$ when $F$ is an embedding by checking the condition for simultaneously solving the systems (\ref{projectioneqns}).

\subsection{Incomplete algorithm for solving embedding}
We now present an incomplete algorithm to solve the local inversion $y=F(x)$ when $F$ is an embedding. The basic idea is to search for an index $i$ within the system of equations among (\ref{projectioneqns}), $F_{i}(x)=y(i)$ which has a polynomially bounded subsequence of its recurring sequence with a minimal polynomial and verify whether the solution obtained for this $i$-th system satisfies all other projection systems. Even if a specific $i$-th projection equation does not satisfy the solution condition of Algorithm 1, it may satisfy the solution obtained from $j$-th system for $j\neq i$. Hence within the bound specified on computation the following algorithm searches for a solution of the embedding or else declares that there was no conclusion. From this discussion we can define an extension of the concept of LC to the case of embedding as the smallest LC of an $i$-th system which results in a solution which satisfies the complete equation (\ref{eqn:fundamental}) of the embedding map. 

\begin{algorithm}\label{Embeddingsolver}
\caption{Incomplete algorithm to find the unique solution of an embedding}
\begin{algorithmic}[1]
\State\textbf{Input}: $F:\ff^n\rightarrow\ff^m$, $n-m=t$, $y$ in $\ff^n$, $M$ an upper bound of polynomial order $O(n^r)$.
\State\textbf{Output}: One solution of $F(x)=y$ if one exists otherwise returns that there is no conclusion.
\Procedure{Solution\_of\_Embedding}{$y=F(x)$}
\State Compute the projection equations (\ref{projectioneqns}) for $i=1,2,\ldots,(t+1)$.
\State $i=1$
\Repeat 
\State Find a solution $x$ of $y(i)=F_i(x)$ for the given bound $M$ using Algorithm 1 in section (\ref{IncompleteAlgo}).
\If {$i$-th projection equation has no conclusion on solution}
\State $i\leftarrow i+1$
\ElsIf {Verify whether 
\[
y(j)=F_j(x)\mbox{ for }j=1,2,\ldots,(t+1),j\neq i
\]
} 
\State Return: solution $x$ if verified for all $j$.
\State Go to \textbf{End} of procedure.
\Else 
\State $i\leftarrow i+1$
\EndIf
\Until $i=(t+1)$
\State Return: No conclusion on solution within the bound.
\State \textbf{End}
\EndProcedure
\end{algorithmic}
\end{algorithm}

\begin{remark}
For an individual $i$-th projection equation there may not be a periodic orbit $S(F_i,y(i))$ nor LC of polynomial order. In such a case the algorithm increments index $i$ of the projection and examines a new system for the periodic solution. However when one of the projection systems has a periodic sequence of recurrence and the LC is within $\lfloor M/2\rfloor$, the algorithm verifies the solution with all other projection systems before returning the solution.
\end{remark}

\subsection{Solving under-determined systems}
When the map $F$ is an embedding, the system $y=F(x)$ is over-determined with larger number of equations than unknown variables to be solved. Hence it is useful to consider the other extreme when number of equations is less than the number of variables to be solved. Hence consider the map $F:\ff^n\rightarrow\ff^m$ where $n>m$. Let $t=n-m$, then if $x$ in $\ff^n$ satisfies the equation there exist assignments to $a=(x_1,x_2,\ldots,x_t)$ in $\ff^t$ such that $F(a,x_{(t+1)},x_{(t+2)},\ldots,x_n)=y$ has a solution. Hence each such assignment $a$ gives rise to a map $F_a:\ff^m\rightarrow\ff^m$ for which the standard theory of local inversion and Algorithm 1 developed in previous section applies. A solution to inversion is then $(a,x)$ for any solution $x$ for $F_a(x)=y$. If $t$ is small enough then such an approach to computation of inversion is feasible. We shall not treat this problem in further detail in this paper except for a special case over the binary fields explained below.

\subsubsection{Under determined systems over the binary field}
We briefly indicate a Boolean equational approach for reducing the variables in the problem of solving $y=F(x)$. This system also represents a Boolean system of equations when the field where the variables and functions in this equation take values is $\ftwo$. Then following the \emph{orthogonal expansion} of functions in a fixed number of variables arising in a subset of equations as shown in \cite{suleImp} these equations can be decomposed into independent systems of equations with non overlapping variables. The systems of equations without common variables can be considered for local inversion independently in reduced number of variables. By subsequent such reductions the original system is brought to a group of sub-systems of equations of the form 
\[
y_i=F_i(x^i)
\]
where $x^i$ denote the variables involved in the $i$-th sub-system which will be an embedding. A solution of the original system then can be obtained by solving each of these systems independently (or parallely) using the local inversion algorithms. Further details and applications of this approach of local inversion for discovering collisions in hash functions shall be developed in forthcoming articles.

The greatest advantage of over defined systems with $F:\ftwo^n\rightarrow\ftwo^m$, $n<m$ is that the forward operation of recurrence for generating the sequences $S(F_i,y(i))$ in the projected systems does not require symbolic or algebraic modeling of the function. The forward operation can be carried out by direct algorithmic description of $F$. On the other hand the difficulty in this case $n>m$ in using the Boolean approach as compared to brute force search over variables is that the decomposition requires that the functions in equation $y=F(x)$ be represented in Boolean function models (in symbolic form). Such models are usually not readily available when the function $F$ is an algorithm involving number theoretic or finite field arithmetic. This is yet another difficult problem of computation. These issues shall be explored in separate articles.

\section{Solution of the key recovery problem in symmetric encryption}
In this section we begin the first application of the incomplete algorithms Algorithm 1, 2 for local inversion. This is the problem of key recovery of symmetric encryption under Known or Chosen Plaintext/Ciphertext Attacks (KPA,CPA, CCA).

\subsection{Block cipher case}
A block cipher is an algorithm $E(.,.)$ which returns the ciphertext block $C$ when the symmetric key block $K$ and the plaintext block $P$ are input to the algorithm
\[
C=E(K,P)
\]
Similarly the decryption algorithm of the block cipher is given by relation
\[
P=D(K,P)
\]
If $P$ is known in the encryption algorithm a function $F_{P}(K)=E(K,P)$ which depends on $P$ defines the fundamental equation (\ref{eqn:fundamental})
\[
C=F_{P}(X)
\]
whose local inversion gives the key $K$. The theory developed in the previous sections then gives us the

\begin{theorem}
Let $K$ and $C$ be both strings in $\ff^n$. If the recurring sequence 
\[
S(F_{P},C)=\{C,F_{P}(C),F_{P}^{(2)}(C),\ldots\}
\]
is periodic and has LC of polynomial order $O(n^k)$ then the key $K$ can be solved by the incomplete algorithm Algorithm 1, in polynomial time using a suitable bound $M$.  
\end{theorem}
 The proof follows from the Theorem 3 of Section 2. The theorem can also be applied for solving the key $K$ using the local inversion of the function $F_{C}(K)=D(K,C)$ at $P=F_{C}(K)$ and an analogous theorem can be stated for this function.

\subsubsection{Variations of the inversion map}
In certain cipher algorithms (such as the older cipher DES) the block lengths of $P$ and $C$ are $64$ while the key length is $58$. Hence in such a case the local inversion maps $F_{P}$ and $F_{C}$ are embeddings and the algorithm applied for local inversion is Algorithm 2. Hence the above theorem is re-written with extension for the case of local inversion of the embedding map $F$ and follows from Theorem 4 of section 3. For instance in this case of inversion problem for DES, the difference $m-n=6$ between number of variables and equations. Hence the number of projection systems defined in (\ref{projectioneqns}) are $7$.

\subsection{Stream cipher case}
A stream cipher is defined by a dynamical system of the type
\beq\label{StreamCipher}
\begin{array}{lcl}
x(k+1) & = & F(x(k)),x(k)\in X\\
w(k)   & = & f(x(k)),k=0,1,2,\ldots
\end{array}
\eeq
where $F$ is the state update map acting in the state space $X=\ff^n$ and $f:X\rightarrow\ff$ is the output map which outputs a stream $w(k)$ for $k=0,1,2,\ldots$. The initial state $x(0)$ is partitioned as $x(0)=(K,IV)$ where $K$ is the symmetric key and $IV$ is the initializing seed (called Initial Vector). Operation of such a stream cipher is carried out as follows:
\begin{enumerate}
    \item Sender and receiver both share the symmetric key $K$ confidentially.
    \item Sender generates an IV and generates the outpiut stream $w(k)$.
    \item Sender encrypts the plaintext stream $p(k)$ as ciphertext stream
    \[
    c(k)=p(k)+w(k)
    \]
    \item Sender sends $(IV,\{p(k)\})$ to the receiver.
    \item Receiver generates $w(k)$ from the $IV$ using the secret key $K$.
    \item Sender decrypts $p(k)=c(k)-w(k)$.
\end{enumerate}

In a known or chosen plaintext attack an adversary has access to a partial stream of plaintext $p(j),j=k_0,k_0+1,\ldots,k_0+m$. Hence the adversary has access to the partial output stream $w(k),k=k_0,k_0+1,\ldots k_0+m$. The problem of cryptanalysis is now to recover $K$ once we are given $IV$ and partial output stream $w(k)$. We may assume that $(m+1)$ is same as number of components in $K$ or the length $l$ of $K$. Then the map to be locally inverted is 
\[
\Phi_{IV}(K):\ff^l\rightarrow\ff^l
\]
where for a known $IV$ and symmetric key input $K$, $\Phi_{IV}(K)=\hat{w}=(w(k_0),\ldots,w(k_0+m))$, where $\Phi_{IV}()$ represents the computation from a given input $K$ in the initial condition $x(0)$ to the output stream from $w(k)$ from $k_0,\ldots,k_0+m$. We then have the following theorem for cryptanalysis of the stream cipher following Theorem 3 of Section 2,

\begin{theorem}
If the recurring sequence 
\[
S(\Phi_{IV},\hat{w})=\{\hat{w},\Phi_{IV}(\hat{w}),\Phi_{IV}^{(2)}(\hat{w}),\ldots\}
\]
is periodic and has LC of polynomial order $O(l^k)$, then the key $K$ can be solved in polynomial time using the incomplete algorithm Algorithm 1 in polynomial time. 
\end{theorem}

The above theorem can be modified suitably if a longer than $l$ output stream $w(k)$ is available for inversion of the embedding $\Phi_{IV}:\ff^l\rightarrow\ff^{(m+1)}$. The details are omitted. 

\subsection{Estimates of bounds on complexities for AES}
AES block cipher is designed in different versions with multiple key sizes hence the estimates of bounds for cryptanalysis of AES by local inversion as well as the maps defined for local inversion are dependent on these key sizes. We shall briefly discuss these issues for each of the cases of versions of AES.

\subsubsection{AES128}
Plaintext $P$ block: $128$ bits, Key $K$ block: $128$ bits, Ciphertext $C$ block: $128$ bits. Number of rounds for processing $P$: $10$, Number of rounds of key schedule: $10$.
\[
\begin{array}{lcl}
\mbox{Encryption function} & : & C=E(K,P)\\
\mbox{Decryption function} & : & P=D(K,C)
\end{array}
\]
Map for local inversion in Known Plaintext Attack (KPA) (or Chosen Plaintext Attack (CCA)), $C=F_{P}(K):=E(K,P)$.
\[
F_{P}(K):\ftwo^{128}\rightarrow\ftwo^{128}
\]
hence $n=128$. Complexity bounds for $n^k$ for polynomial time search for local inversion are shown below. With $40$ bit brute force search, the map becomes an embedding with $n=88$.
\newline

\begin{tabular}{||l|l|l||}
\hline
  $k$  & $n^k$ & $n=88$.  \\
\hline
3 & 2 million & 0.7 million\\
2 & 16 thousand & 7744\\
\hline
\end{tabular}

\subsubsection{AES192}
Plaintext block $P$: $128$ bits, Key block $K$: $192$ bits, Ciphertext block $C$: $128$ bits. Number of rounds of plaintext processing: $12$, Key schedule rounds: $8$.
\[
\begin{array}{lcl}
\mbox{Encryption function} & : & C=E(K,P)\\
\mbox{Decryption function} & : & P=E(K,C)
\end{array}
\]
Map for local inversion requires two plaintext blocks $P_1$, $P_2$ encrypted by same key to ciphertext blocks $C_1$, $C_2$.
\[
F_{(P_1,P_2)}(K):\ftwo^{192}\rightarrow\ftwo^{256}
\]
hence the map is an embedding and $n=192$. With $40$ bit brute force search $n=152$
\newline

\begin{tabular}{||l|l|l||}
\hline
  $k$  & $n^k$ & $n=152$.  \\
\hline
3 & 7 million & 3.5 million\\
2 & 36 thousand & 23 thousand\\
\hline
\end{tabular}

\subsubsection{AES256}
Plaintext block $P$: $128$ bits, Key block $K$: $256$ bits, Ciphertext block $C$: $128$ bits. Number of rounds for processing $P$: $14$, Number of rounds in key schedule: $7$.

KPA/CPA: Two plaintext blocks encrypted by one key block.
\[
\begin{array}{lcl}
\mbox{Encryption function} & : & (C_1,C_2)=E(K,(P_1,P_2))\\
\mbox{Decryption function} & : & (P_1,P_2)=D(K,(C_1,C_2))
\end{array}
\]
Map for local inversion
\[
F_{(P_1,P_2)}(K):\ftwo^{256}\rightarrow\ftwo^{256}
\]
hence $n=256$. With $40$ bit brute force search the map is an embedding with $n=216$.
\newline

\begin{tabular}{||l|l|l||}
\hline
  $k$  & $n^k$ & $n=216$.  \\
\hline
3 & 16 million & 10 million\\
2 & 65 thousand & 46 thousand\\
\hline
\end{tabular}
\section{Cryptanalysis of RSA}
In this section we investigate the problem of cryptanalysis of RSA formulated as a local inversion problem. In general, in any public key encryption scheme, with private key $R$, public key $U=H(R)$, Encryption function $c=E(U,m)$ and the decryption function $m=D(R,c)$ for any known ciphertext $c$ and pair $(m,c)$, local inversion of maps $F_{U}$ and $F_{c}$ defined as
\[
\begin{array}{lclcl}
c & = & F_{U}(m) & = & E(U,m)\\
m & = & F_{c}(R) & = & D(R,c)
\end{array}
\] 
solve the unknown plaintext input $m$ and private key $R$. Private key can also be solved by inverting the map $U=H(R)$. 

\subsection{Inversion problems in RSA}
RSA has the private parameters $p,q,d$, $p,q$ unequal odd primes, public parameters $n=pq$ called the \emph{modulus} and an exponent $e$ which defines the private parameter $d$ such that $ed=1\mod\phi(n)$. Let the modulus has bit length $l=\log(n)+1$. We have the following local inversion problems.

\subsubsection{Breaking RSA by ciphertext inversion without factoring $n$}
The map in this case is $F_{U}$ which is obtained from the relation
\[
c=m^{e}\mod c=F_{e}(m):\ff_{2^l}\rightarrow\ff_{2^l}
\]
For any number $a$ in $\zz_n$ denote by $(a)=(a_0,a_1,\ldots,a_l)$ the binary string in the binary expansion of $a$. Similarly let $[(a)]$ denote the number in $[0,n-1]$ corresponding to the binary $l$-tuple $(a)$. Define the map $F_{e}:\ftwo^{l}\rightarrow\ftwo^{l}$
by
\[
(y)=F_{e}(x):=(x^e\mod n)
\]
The dynamical system (\ref{eqn:dynsys}) defined by this map in $\ftwo^l$ is
\beq\label{RSAdyn1}
(y_{(j+1)})=(y_j^e\mod n), j=1,2,\dots
\eeq
where $y_1=(c)$. Hence this generates the sequence 
\[
S(F_{e},c)=\{(c),(c^e\mod n),(c^{e^2}\mod n),\ldots\}
\]
Now note that this map $F_{e}$ has unique inverse $x=m$ in $[0,n-1]$ for any given $c$ in $[0,n-1]$ because of the arithmetic of RSA.
Then from Theorem 3 we get the following theorem on cryptanalysis of RSA for recovering plaintext $m$ given ciphertext $c$ without factoring the modulus

\begin{theorem}
If the sub-sequence 
\[
S(F_{e},c)=\{(c),(c^e\mod n),(c^{e^2}\mod n),\ldots,(c^{e^M}\mod n)\}
\]
of the recurrence of system (\ref{RSAdyn1})
is given upto $M$ terms where $M$ is of polynomial order $O(l^r)$ and has LC of polynomial order $m\leq\lfloor M/2\rfloor$, then the plaintext $m$ can be recovered in polynomial time by from the sequence by Algorithm 1.
\end{theorem}

\begin{proof}
The sequence $S(F_{e},c)$ is periodic. If the subsequence is given upto $M$ terms of polynomial order and the minimal polynomial of the full sequence exists then from Theorem 3 it follows that the minimal polynomial and the local inverse can be computed in polynomial time using Algorithm 1. Hence the plaintext $x=m$ is found such that $c=x^e\mod n$. 
\end{proof}

\begin{remark}
Note that this plaintext recovery using local inversion of $F_{e}$ does not utilize factorization of $n$ in any way. On the other hand since the solution $m$ for a given $c$ is unique the sequence $S(F_{e},c)$ is periodic. Hence there is always a minimal polynomial for the sequence and a solution of the inverse of $c$.
\end{remark}

The well known cycling attack on RSA \cite{GysSeb} determines the period of the sequence $S(F_{e},c)$. However this attack is not feasible because the period of this sequence is exponential. While the LC of the sequence might turn out to be of polynomial size. Hence the local inversion attack has a chance of success in certain cases.

In view of the above remarks we state

\begin{corollary}
RSA map $c=m^e\mod n$ can be reversed without factoring the modulus $n$.
\end{corollary}

\begin{proof}
Follows from Theorem 3 since the sequence $S(F_{e},c)$ is periodic there exists the minimal polynomial which computes the local inverse as the last element of the periodic sequence. Hence decryption is achieved without factoring.
\end{proof}

\subsubsection{Breaking RSA using CCA for any encryption using the same private keys without factoring $n$}
This is another local inversion attack possible on RSA to decrypt any ciphertext without factoring the modulus as long same private keys are used for encryption. Consider the function $F_{c}$ defined above used for CCA, using the decryption function of RSA.
\[
m=F_{c}(d):=c^d\mod n
\]
Note that the unknown in this map is the private key $d$ which belongs to $\zz_{\phi(n)}$. Since in general 
\[
\phi(n)\leq (0.6)n
\]
The number of bits for the domain of the map to be inverted which is $\zz_{\phi(n)}$ can be chosen as $l$ same as that of length of $n$. Define the map $F_{c}:\ff_{2^l}\rightarrow\ff_{2^l}$ on the binary representations $(x)$ of numbers in $\zz_n$ and $\zz_{\phi(n)}$ by
\[
(y)=F_{c}((x)):=(c^{[(x)]}\mod n)
\]
The dynamical system generated by this map is
\beq\label{RSAdyn2}
(y_{(j+1)})=(c^{y_j}\mod n),j=1,2,\ldots
\eeq
with $y_1=(m)$. The sequence generated by the map is then
\[
\{y_1=(m),y_2=(c^m\mod n),y_3=(c^{[y_2]}\mod n,\dots\} 
\]
Note that the map $F_c(.)$ does not have GOE. Hence all such sequences are periodic. We can now state,

\begin{theorem}
Consider the CCA on RSA giving the pair $(m,c)$ to the attacker. If the sub-sequence of the sequence of recurrences of the system (\ref{RSAdyn2})
is given upto $M$ terms where $M$ is of polynomial order $O(l^r)$ and has LC of polynomial order $m\leq\lfloor M/2\rfloor$, then the Algorithm 1 computes the decryption of any $c$ obtained using the same public keys in polynomial time.
\end{theorem}

\begin{proof}
Conditions of the theorem grant that that the subsequence $S(F_{c},m)$ given upto $M$ terms has a minimal polynomial of polynomial order. Hence from theorem 3 we get the unique inverse of $m$ in periodic the sequence of recurrence of the system (\ref{RSAdyn2}). Let the period of the sequence be $N$. Then $[(y_{N})]=m$ and
\[
m=c^{[y_{(N-1)}]}\mod n
\]
From the RSA decryption relation it follows that
\[
c=m^e\mod n=c^{e[y_{(N-1)}]}\mod n
\]
Which implies $e[y_{(N-1)}]=1\mod\phi(n)$ and hence for any other ciphertext $\tilde{c}$ created from the same public keys $e,n$. 
\[
\tilde{m}=c^{[y_{(N-1)}]}\mod n
\]
decrypts $\tilde{c}$.
\end{proof}

\begin{remark}
Although the actual inverse $y_{(N-1)}$ need not equal to $d$, the decryption works as long as same private keys are used in the encryption and $d$ for decryption. Hence once the local inversion is successful for one pair $(m,c)$ any other ciphertext using the same private keys is decrypted.
The theorem shows that the local inversion does not involve factoring the modulus.
\end{remark}

In view of the above remarks we state

\begin{corollary}
RSA can be decrypted in a CCA given $(m,c)$ using local inversion of the map $F_{c}$ without factoring the modulus $n$ and the resultant local inverse can decrypt any ciphertext obtained by the same public keys.
\end{corollary}

\begin{proof}
In any CCA, the sequence $S(F_{c},m)$ is periodic. Hence the sequence has a minimal polynomial which allows local inversion of the map $F_{c}$ to find $x$ such that $m=F_{c}(x)$. The inverse then satisfies $m=c^x\mod n$. 
\end{proof}

\subsubsection{Estimates of LC for feasible solution}
As per above theorems on polynomial time solvability of the local inversion problems, the estimates of sizes of linear systems to be solved or the LC are calculated for standard sizes of lengths of RSA modulus. Following table shows the estimates of largest sizes of linear systems required to be solved in the Algorithm 1.
\newline

\begin{tabular}{||l|l|l|l||}
\hline
$l$-length of $n$ & $k$ & $M$ & $m$ or the LC\\
\hline
1024 & 3 & 1 Billion & 537 Million\\
1024 & 2 & 1 Million & 524,288 \\
2048 & 3 & 8.5 Trillion & 4.3 Trillion\\
2048 & 2 & 4.2 Million & 2.1 Million\\
\hline
\end{tabular}

\subsubsection{Factoring as local inversion}
The third way to attack RSA by local inversion is to invert the map which generates the public key from the private key. This is the factorization problem of computing factors of $n$. Analysis of LC of this map for solving factorization is yet another interesting application which shall be investigated in a separate article. Most previously known approaches to factorization are based on Number Theory. Local inversion approach shall provide a new approach to factoring.
\section{Solution of Elliptic Curve Discrete Log Problem (ECDLP) as Local inversion}
This section considers another important problem of cryptanalysis that of solving the ECDLP. In this problem an elliptic curve $E(\ff_q)$ is given over the field $\ff_q$ of char not equal to $2,3$. The Weierstass normal equation of $E$ is
\[
y^2=x^3+Ax+B
\]
where $A,B$ are given in $\ff_q$. The points on $E$ are chosen with co-ordinates $(x,y)$ in $\fq$. Another class of curves used in practice are Koblitz curves $E(\ff_{2^n})$.

In the discrete log problem there are given points $P$ and $Q=[m]P$ in $E$ where $m$ is the integral multiplier. It is required to solve for the multiplier $m$. Define the map
\[
F_{P}:m\mapsto[m]P
\]
then the local inversion of $Q=F_{P}(m)$ solves the ECDLP. However the map as described needs to be expressed in the standard form of the local inversion problem described earlier so that the condition for feasible solution of inversion of Theorem 3 can be utilized.

\subsection{Formulation as local inversion}
The multiplier $m$ is less than the order of the cyclic group $n=<P>$ in $E$. Sometimes the group $E$ itself has prime order $n=\sharp E$ hence the order of $<P>$ is $n$. Hence to fix the number of bits in $m$ we consider estimates of the order of $E$. The well known bound on the order of $E$ is
\[
\sharp E(\fq)\leq q+1+2\sqrt{q}
\]
Assuming $q>4$ we have $\sharp E\leq 2q$. On the other hand the point $Q$ in $E$ has two co-ordinates in $\fq$. Thus the bit length of $\sharp E\leq 1+\log (2q)=2+\log q=l$ while the bit length of two co-ordinates of a point taken together is $1+\log (2q)=l$. Consider the map $F_{P}$ defining the scalar multiplication of $P$ in $E$ 
\[
\begin{array}{lcl}
F_{P} & : & \zz_n\rightarrow E\\
 & & m\mapsto [m]P
\end{array} 
\]
In the binary co-ordinate expansion on both sides this mapping is  
\beq\label{mapFPm}
F_{P}((m))=((Q_x),(Q_y))
\eeq
where $(m)$ denotes the co-efficients in the binary expansion of $m$ and $(Q_x)$, $(Q_y)$ are co-efficients in the binary expansions of co-ordinates of $Q=[m]P$. We shall denote the binary expansion of the co-ordinate pair of $Q$ as $(Q)$. 

\subsubsection{Formulation of $F_{P}$ as a map over the binary field}
In order to utilize the previous theory on local inversion of maps we need to formulate the map (\ref{mapFPm}) as a map in the cartesian spaces of $\ftwo$ and understand whether it is an embedding. 

Let $r=1+\log n$ where $n$ is the order of $<P>$. Then $F_{P}$ in (\ref{mapFPm}) represents a map $F_{P}:\ftwo^r\rightarrow\ftwo^l$ where $r<l$. Thus $F_{P}$ is an embedding. Hence it is required to apply Algorithm 2 to solve the embedding equation for the local inversion of $F_{P}(m)=Q$ from the binary representation in (\ref{mapFPm}). The application of Algorithm 2 requires that $n$ the order of $P$ is known. Let $t=l-r$, then the projection equations (\ref{projectioneqns}) give $(t+1)$ standard equations 
\[
\Pi_i\circ F_{P}(x)=\Pi_i((Q)),i=1,2,\ldots (t+1)
\]
denote by $F_i=\Pi_i\circ F_{P}$ and $y(i)=\Pi_i((Q))$ where $\Pi_i$ are projection on $r$ components of $y$ as defined in (\ref{projectioneqns}). Following Theorem 3 now we have

\begin{theorem}
If for any of the indices $i$, $1\leq i\leq (t+1)$ the projection equation $F_{i}(x)=y(i)$ has a periodic recurrence sequence $S(F_{i},y(i))$ with a LC $m\leq\lfloor M/2\rfloor$ where $M$ is of polynomial order $O(r^k)$ and the local inverse $x$ satisfies all other projection equations then the ECDLP is solved in polynomial time by Algorithm 1.
\end{theorem}

\begin{proof}
For each of the projections $\Pi_i$ the equations $F_i(x)=y(i)$ have the map $F_{i}:\ftwo^r\rightarrow\ftwo^r$. Hence if the recurrent sequence defined by the system (\ref{eqn:dynsys}) for $F_i$
\[
y(i)_{(j+1)}=F_i(y(i)_j),j=0,1,2,\ldots
\]
where $y(i)_0=\Pi_i((Q))$, is periodic and has a minimal polynomial of deg $m$ then the local inverse $x$ satisfying $F_i(x)=y(i)$ can be computed in polynomial time. Since this is a unique solution in the periodic orbit of the recurrence sequence $S(F_i,y(i))$, if this solution also satisfies all other projection equations, then $x$ satisfies the equation (\ref{mapFPm}) of the embedding. Hence the ECDLP is the local inverse $x$ which is solvable in polynomial time.  
\end{proof}

The sizes of fields $\fq$ for practically used elliptic curves are close to about $256$ bits such as for instance in the current bitcoin curve $\mbox{secp}256k1$. For such sizes of $q=2^{256}$ and elliptic curves defined on them we have $l=258$. Following table gives sizes of $k$, $M$ and $m$ for the elliptic curve $\mbox{secp}256k1$.
\newline

\begin{tabular}{||l|l|l|l||}
\hline
$n<$ & $k$ & $M<$ & $m<$\\
\hline
 258 & 3 & 17 Million & 8.5 Million\\ 258 & 2 & 66564 & 33282\\
 \hline
\end{tabular}
\\
The actual numbers of bits for the order $n$ and bounds $M$ and linear complexity $m$ shall depend on actual value of $n$. Above table shows bounds for the curve $\mbox{secp}256k1$.

\subsection{Discrete logarithm over prime fields as local inversion}
The discrete logarithm computation over the multiplicative group of finite fields has also been an important problem in cryptography called in short as Discrete Log Problem (DLP). In a prime field $\fp$ the DLP is concerned with computation of the index $x$ such that
\[
b=a^x\mod p
\]
where $a,b$ are given in $[1,p-1]$ and $x$ belongs to $[0,p-2]$. $x$ is called the \emph{Discrete Log} of $b$ with respect to base $a$. Since 
\[
a^x\mod p=a^{(x\mod (p-1)}\mod p
\]
by Fermat's theorem. Consider the map
\[
F_a:\fp\rightarrow\fp
\]
such that $F_a(0)=F_a(p-1)=1$. For all other $x$ the map is as defined above. Then the GOE of $F_a()$ is the only point $0$ in $\fp$ since there does not exist an $x$ in $[0,p-2]$ such that $F_a(x)=0$. Hence only $F_a(x)=1$ has two solutions $x=0$ and $x=p-1$. For $b\neq 1$ the local inverse is unique in a periodic orbit. The sequence $S(F_a,b)$ is given by 
\beq\label{Fpseq}
S(F_a,b)+\{b,a^b\mod p,a^{(a^b\mod p)}\mod p,\ldots\}
\eeq
where the recurrence formula is given by
\beq\label{Fprecurrence}
\begin{array}{lcl}
  F_a(b) & = & a^b\mod p  \\
  F_a^{(2)}(b) & = & a^{a^b\mod p}\mod p\\
  F_a^{(k+1)}(b) & = &
  a^{F_a^{(k)}(b)}\mod p
\end{array}
\eeq
Thus for $b\neq 1$ the sequence $S(F_a,b)$ is periodic. Hence following Theorem 2 we have

\begin{corollary}
For $b\neq 1$ every sequence $S(F_a,b)$ (\ref{Fpseq}) in a prime field has a minimal polynomial (\ref{minpoly})
\[
m(X)=X^{(m)}-\sum_{i=0}^{(m-1)}\al_iX^i
\]
over $\fp[X]$ with $\al_0\neq 0$ and the DL of $b$ with respect to base $a$ is the local inverse (\ref{solution})
\[
x=(1/\al_0)F_a^{(m-1)}(b)-\sum_{i=1}^{(m-1)}\al_iF_a^{(i-1)}(b)
\]
\end{corollary}

Proof follows from the definition of the recurrence map defined in (\ref{Fprecurrence}) and Theorem 2. Let $l=\mbox{length}\;p$. Then from Theorem 3 we have

\begin{corollary}
If the subsequence of $S(F_a,b)$ given upto polynomial number of terms $M=O(l^k)$ has a minimal polynomial of degree $m\leq\lfloor M/2\rfloor$ over $\fp[X]$ then the DL of $b$ with respect to base $a$ can be computed in polynomial time. 
\end{corollary}

\subsection{Discrete logarithm over binary extension fields by local inversion}
The case of defining the map for computing DL as local inverse in extension fields is now considered. Consider the field $\ff=\ff_{2^n}$ and let $a$ be a primitive element of $\ff^*$. Then the exponent equation over $\ff$ is
\[
b=a^x
\]
where $x\in [0,2^n-1]$ and $a,b$ are in $\ff^*$. In a polynomial basis for $\ff$, $a,b$ have $n$-bit representation, while the index $x$ also requires $n$-bits. Hence exponentiation can be formulated as the map
\beq\label{Extfieldexponent}
\begin{array}{lcl}
F_{a} & : & \ftwo^n\rightarrow\ftwo^n\\
 & & (x)\mapsto (a^{[(x)]}) 
\end{array}
\eeq
where $(x)=(x_0,\ldots,x_{(n-1)})$ is a binary representation of the index $x$ in $[0,2^n-1]$, $[(x)]$ is the reversal of a binary $n$-bit string $(x)$ as a number in $[0,2^n-1]$, $(a)$ denotes the $n$-bit representation of $a$ in $\ff$ in a fixed polynomial basis. Hence the recurrence can now be defined starting from $b$ in $\ff$ by relations
\beq\label{Recinextfield}
\begin{array}{lcl}
F_a(b) & = & (a^[(b)])\\
F_a^{(k+1)}(b) & = & (a^{[(F_a^{k}(b))]})
\end{array}
\eeq
The map $F_a$ has GOE the only point $(0)$ (the zero in $\ff$). Every point $b\neq 1$ has a periodic orbit. Hence we can now have an identical corollary in the extension field case for computation of the DL.

\begin{corollary}
If the subsequence of $S(F_a,b)$ given upto polynomial number of terms $M=O(n^k)$ has a minimal polynomial of degree $m\leq\lfloor M/2\rfloor$ over $\ftwo[X]$ then the DL of $b$ with respect to base $a$ can be computed in $\ff_{2^n}$ in polynomial time. 
\end{corollary}

Proof follows from definitions of the exponent map (\ref{Extfieldexponent}), the recurrence (\ref{Recinextfield}), Theorem 2 and Theorem 3. 

The case of DL computation over extension fields of other characteristics can be discussed on similar lines after defining the exponent map and the recurrence relation. These developments are omitted.
\section*{Appendix}
This section presents proofs of previous results stated in \cite{suleLI1} which are extended to work for any finite field $\ff$.

\subsubsection*{Proof of Theorem (\ref{strofsolns})}.
\begin{proof}
From the definition of the periodic orbits and chains it follows that the space $\ff^n$ is partitioned by the action of $F$ in periodic orbits and segments of chains from points in GOE to a point which is mapped by $F$ to a periodic orbit. Hence given any point $y$ it is either on a unique chain segment or on a periodic orbit. If there is no solution to $F(x)=y$ then $y$ belongs to GOE conversely if $y$ is in GOE then there is no $x$ such that $F(x)=y$. If $y$ is in a periodic orbit $P$, there is unique predecessor $x$ in $P$ such that $F(x)=y$. Any other solution $x$ which is outside $P$ cannot be in any other periodic orbit since the two orbits cannot intersect at $y$. Hence all other solutions $x$ are on chains merging with $P$ at $y$ under iterations of $F$. Hence every solution outside $P$ is on one of the chains $F^{(k)}(z)$ for $z$ in GOE and $k\geq 1$. 
\end{proof}

\subsubsection*{Proof of Theorem (\ref{Solutioninperiodicorbit})}.
\begin{proof}
Let $S(F,y)$ be periodic of period $N$. The point $x=F^{(N-1)}(y)$ then satisfies the equation $F(x)=y$, hence this is one and the unique solution of the equation in the periodic orbit. But then it follows that for this solution $x$, the periodic sequence $S(F,x)=S(F,y)$. Hence if $m(X)$ as described in (\ref{minpoly}) is the minimal polynomial of $S(F,Y)$,
\[
m(X)(x)=F^{(m)}(x)-\sum_{i=0}^{(m-1)}\al_iF^{i}(x)=0
\]
From this expression of $m(X)(x)$ the term $x$ can be solved uniquely since $\al_0\neq 0$. 
\[
x=(1/\al_0)[F^{(m)}(x)-\sum_{i=1}^{(m-1)}\al_iF^{(i)}(x)]
\]
Then by using the condition $y=F(x)$ one gets the relation (\ref{solution}). This is the expression of the unique solution $x$ in the periodic orbit $S(F,y)$.
\end{proof}
\section{Conclusions}
It is shown that cryptanalysis of symmetric as well as public key primitives can be formulated as a problem of local inversion of a map $F:\ff^n\rightarrow\ff^n$ at $y$ in $\ff^n$ given the equation $y=F(x)$ over a finite field $\ff$. In some cases $F$ is an embedding and maps $\ff^n$ to $\ff^m$ for $n<m$. An incomplete algorithm is developed which can address both of these problems especially in the practical situation, when only a polynomial size length of the sequence of recurrence defined by $F$ and $y$ can be made available for inversion. If the linear complexity (LC) of the recurrence defined by $F$ and $y$ is of polynomially bounded order $O(n^k)$, then a possible  solution can be computed in polynomial time. Local inversion is shown as a uniform methodology for cryptanalysis of symmetric encryption algorithms, cryptanalysis of RSA without factoring the modulus and computation of the discrete logarithms on elliptic curves and finite fields. In real life cases it will be worthwhile to carry out the computations proposed in the incomplete algorithm to determine density of such low LC instances occurring within the bounds of practically feasible time and memory. Hence such computations can be used for estimating security of cryptographic primitives as functions of their parameters in terms of sizes of linear systems required to be solved over specific finite field. Main conclusion of this paper is that the low LC of recurrence of a map $F$ on a value $y$ defined by the cryptographic primitive is the vital sufficient condition for its weakness. Hence computation of LC of recurrences of cryptographic primitives should be considered for standardization of cryptographic algorithms. 

\subsubsection*{}
\begin{center}
    Acknowledgements
    
\noindent
Author is thankful to Shashank Sule for a useful discussion and to Ramachandran and Shravani Shahapure for help in correcting errors.
\end{center}

\end{document}